\documentclass[12pt]{article}
\usepackage{fullpage}
\usepackage[utf8]{inputenc}

\usepackage{amsfonts}
\usepackage{amsmath}

\usepackage{amsthm}
\usepackage{amssymb}
\usepackage{graphicx}
\usepackage{microtype}
\usepackage{wrapfig}
\usepackage{subfig}
\usepackage{array}
\usepackage{tabularx}
\usepackage{xparse}
\usepackage{algorithm}
\usepackage[noend]{algpseudocode}

\usepackage{appendix}

\floatname{algorithm}{Algorithm}

\newcommand{\tuple}[1]{\langle #1 \rangle}
\newcommand{\nat}{\mathbb{N}}

\newcommand{\ceil}[1]{\left \lceil #1 \right \rceil }

\newtheorem{theorem}{Theorem}
\newtheorem{lemma}[theorem]{Lemma}
\newtheorem{definition}{Definition}
\theoremstyle{plain}
\newtheorem{obs}{Observation}

\newtheorem{example}{Example}
\newtheorem{remark}{Remark}

\newtheorem{conjecture}{Conjecture}
\newtheorem{question}{Question}

\title{On vertex coloring without monochromatic triangles}

\author{
Micha\l~Karpi\'nski\thanks{karp@cs.uni.wroc.pl}, Krzysztof Piecuch\thanks{kpiecuch@cs.uni.wroc.pl}\\[2mm]
Institute of Computer Science \\
University of Wroc\l aw \\
Joliot-Curie 15, 50-383 Wroc\l aw, Poland \\
}

\begin{document}

    \maketitle

    \begin{abstract}
      We study a certain relaxation of the classic vertex coloring problem,
      namely, a coloring of vertices of undirected, simple graphs,
      such that there are no monochromatic triangles. We give the first classification
      of the problem in terms of classic and parametrized algorithms. Several
      computational complexity results are also presented, which improve on the previous results found in the literature.
      We propose the new structural parameter for undirected, simple graphs -- the triangle-free chromatic number $\chi_3$.
      We bound $\chi_3$ by other known structural parameters. We also present two classes of graphs with interesting
      coloring properties, that play pivotal role in proving useful observation about our problem.
      We give/ask several conjectures/questions throughout this paper
      to encourage new research in the area of graph coloring.
    \end{abstract}

  \section{Introduction}

  Graph coloring is probably the most popular subject in graph theory.
  It is an interesting topic from both algorithmic and combinatoric points of view.
  The coloring problems have many practical applications in areas such as operations research,
  scheduling and computational biology. For a recent survey one can turn to \cite{survey}.
  In this paper we study a variation of the classic coloring -- we call it the {\em triangle-free coloring}
  problem, in which we ask for an assignment of colors to the vertices of a given graph,
  such that the number of colors used is minimum and that each cycle of length 3 has at least two vertices
  colored differently. We show that our problem has interesting graph-theoretical properties and we also present some
  evidence that this problem might be easier than the classic vertex coloring. This suggests that
  studying our variation, new results can be achieved in the field of classic vertex coloring,
  which is known to be one of the hardest known optimization problems. Apart from theoretical
  motivation, there is also a practical one -- vertex coloring without monochromatic cycles
  can be used in the study of consumption behavior \cite{motivation}.

  \subsection{Related work}

  Some researchers have already considered coloring problems that are similar to our variation.
  The class of planar graphs has been of particular interest, for example, Angelini and Frati \cite{angelini}
  study planar graphs that admit an acyclic 3-coloring -- a proper coloring in which 
  every 2-chromatic subgraph is acyclic. Algorithms for acyclic coloring
  can be used to solve/approximate a triangle-free coloring, although we
  do not explore this possibility in this paper.
  Another result is of Kaiser and \v{S}krekovski \cite{kaiser}, where
  they prove that every planar graph has a 2-coloring such that no cycle of length 3 or 4 is monochromatic.
  Thomassen \cite{2listplanar}, on the other hand, considers list-coloring of planar graphs without monochromatic triangles.
  Few hardness results for our problem are known -- Karpi\'nski \cite{karpinski} showed that verifying whether
  a graph admits a 2-coloring without monochromatic cycles of fixed length is $\mathcal{NP}$-complete.
  His proof was then simplified by Shitov \cite{shitov}, who also proposed and proved the hardness of
  an extension of our problem, where additional restriction is imposed on the coloring in the form of
  the set of {\em polar} edges -- edges that must not be monochromatic in the resulting coloring.
   
  \subsection{Our contribution}

  Several novel results are presented in this paper.
  First, we explore the graph-theoretical side of our problem. We propose the new
  structural parameter $\chi_3(G)$ which is the minimum number of colors needed
  to label the vertices of an undirected, simple graph $G$, such that there are no
  monochromatic triangles. We then bound this new parameter by $\omega(G)$, $\chi(G)$ and
  $\Delta(G)$, which are clique number, chromatic number and the largest vertex degree of $G$, respectively.
  We conjecture that $\chi_3(G) \leq \lceil \omega(G)/2 \rceil + 1$ and we construct an infinite class
  of graphs for which this upper bound is met. In our construction we use what we call {\em cycle-cliques}
  as building blocks. Those gadgets have interesting coloring properties and are also used in proving hardness results
  later in the paper.

  For the positive side, several known graph classes are presented, for which our problem can be solved efficiently, for
  example, we can find $\chi_3$ on planar graphs in polynomial time, whereas finding $\chi$ is $\mathcal{NP}$-complete,
  even on planar graphs with maximum degree 4 \cite{4regplanar}. We use the fact that when $\chi$ is small (less than 5),
  then we can reduce our problem to the problem of deciding if a given graph is triangle-free. In general, the time needed for listing
  all triangles of a graph is $O(m \sqrt{m})$ \cite{triangle-listing}, but in the presented graphs,
  we can find if a graph is triangle-free in time $O(n)$. We also prove that our problem is fixed-parameter
  tractable, when the parameter is the vertex cover number.

  We present several hardness results, which improve on the work of Karpi\'nski \cite{karpinski} and Shitov \cite{shitov}.
  We show that given any fixed number $q \geq 2$, determining if graph is triangle-free
  $q$-colorable is $\mathcal{NP}$-hard. This improves on the result given in \cite{karpinski}, where the author shows hardness
  only for $q = 2$. Another improvement to Karpi\'nski's result \cite{karpinski} is $\mathcal{NP}$-hardness proof of triangle-free
  2-coloring problem for the graphs that does not contain clique of size 4 as a subgraph. In \cite{shitov}, author
  formulates and proves the $\mathcal{NP}$-hardness of triangle-free 2-coloring problem where additional set of polar
  edges -- edges that must not be monochromatic -- are given on input. We show that this variation of our problem remains
  $\mathcal{NP}$-hard, even on graphs with maximum degree 3 -- the sub-cubic graphs.

  Throughout the paper we ask many questions and propose conjectures,
  which we hope will spark an interest in this new variant of vertex coloring.
  
  \subsection{Structure of the paper}

  In Section 2 we give definitions and notations
  used throughout the paper, as well as formulation of all considered problems.
  We give several bounds on $\chi_3$ in Section 3 and we construct
  an infinite class of graphs for which one of the conjectured bounds is tight. In Section 4 we present
  efficient algorithms for the triangle-free coloring problem for certain classes of graphs. A single FPT result
  are also shown there. Hardness results are presented in Section 5. We finish the paper with some concluding remarks in
  Section 6, where we also present the reader several open problems.
  
  \section{Preliminaries}

  We recall several properties and notions from graph theory that are of interest to us. Let $G=(V,E)$ be a finite,
  undirected, unweighted, simple graph with the vertex set $V$ and the edge set $E$. Set of vertices
  adjacent to some vertex $v \in V$ is called the neighborhood of $v$ and is denoted by $N_G(v)$. The degree
  of a vertex $v$ is defined as $d_{G}(v)=|N_G(v)|$ (we omit index $G$ in the
  notation if it's clear from the context which graph is considered). The degree of $G$ is
  $\Delta(G) = \max_{v \in V} d(v)$. We write that $H$ is a subgraph of $G$ as $H \subseteq G$, and
  $H$ is isomorphic to $G$ as $H \sim G$. We say that graph $G$ is $H$-free if there is no $H' \subseteq G$ such that $H' \sim H$.
  
  An elementary contraction of graph $G$ is obtained by {\em identification} of (ordered) pair of vertices $\tuple{u,v}$.
  This process is performed by removing $v$ and adding an edge $uw$, for every $w \in N(v)$. After that, multi-edges and loops
  are removed.

  Let $K_n$ ($n \in \nat$) be the clique of size $n$. The greatest integer $r$
  such that $K_r \subseteq G$ is the {\em clique number} $\omega(G)$ of $G$. A subset $I$ of vertices $V$ is
  an independent set, if none of its vertices are adjacent. A subset $W$ of vertices $V$ is called a vertex cover,
  if every edge in the graph is incident to at least one vertex of $W$. The smallest integer $r$ such that $|W|=r$ and
  $W$ is a vertex cover of $G$, is called the {\em vertex cover number}.
  
  In this paper we study coloring of vertices, and to avoid confusion, we distinguish two types of coloring.
  A {\em classic k-coloring} of a graph is a function $c : V \rightarrow \{1,\dots, k\}$, such that there are no two
  adjacent vertices $u$ and $v$, for which $c(u)=c(v)$. Given $G$, the smallest $k$ for which there exists a classic $k$-coloring
  for $G$ is called the {\em chromatic number} and is denoted as $\chi(G)$. A {\em triangle-free k-coloring} of a graph is a function
  $c : V \rightarrow \{1,\dots, k\}$, such that there are no three mutually adjacent vertices $u$, $v$ and $w$, for which $c(u)=c(v)=c(w)$.
  If such vertices exist, then the induced subgraph ($K_3$) is called a {\em monochromatic} triangle.
  Given $G$, the smallest $k$ for which there exists a triangle-free $k$-coloring for $G$ we call the {\em triangle-free chromatic number} and
  we denote it as $\chi_3(G)$.

  We now formulate the decision problems investigated in this paper, for any fixed integer $q>0$:

  \vspace{0.3\baselineskip}
    
  \noindent \textsc{TriangleFree}-q-\textsc{Coloring}

  \noindent {\em Input:} A finite, undirected, simple graph $G$.

  \noindent {\em Question:} Is there a triangle-free $q$-coloring of $G$?

  \vspace{0.3\baselineskip}

  \noindent \textsc{TriangleFreePolar}-q-\textsc{Coloring}

  \noindent {\em Input:} A finite, undirected, simple graph $G=(V,E)$ and a subset $S \subseteq E$.

  \noindent {\em Question:} Is there a triangle-free $q$-coloring of $G$, such that no edge in $S$ is monochromatic?
  
  \vspace{0.3\baselineskip}

  Fixed-parameter tractable (FPT) algorithm (w.r.t. parameter $k$) is an algorithm for whose running time is polynomial for
  any fixed value of $k$, i.e., running time is $O(f(k)\cdot n^c)$, where $n$ is a size of the input, $c$ is a constant and
  $f$ is a computable function.
  
  We denote $\ceil{x}$ to be the smallest integer not less than $x$. Let $r \in \nat$,
  we define a binary operator $+_r$ on the set $\mathbb{Z}_r=\{0,1,2,\dots,r-1\}$,
  which is addition modulo $r$, i.e., for every $n,m \in \mathbb{Z}_r$, $n+_rm= n+m$ $(\text{mod} \,\, r)$. The
  operator $-_r$ can be defined in a similar way. 
  
  \section{Bounds on the triangle-free chromatic number}

  We first give simple bounds on $\chi_3(G)$ in terms of $\omega(G)$ and $\chi(G)$:

  \begin{theorem}\label{thm:bounds}
    For any graph $G$: $\ceil{\frac{\omega(G)}{2}} \leq \chi_3(G) \leq \ceil{\frac{\chi(G)}{2}}$.
  \end{theorem}

  \begin{proof}
    To see that the lower bound holds, take any clique in $G$ of maximum cardinality.
    We can use one color for at most two vertices of that clique, otherwise we would create a monochromatic triangle.
    Therefore we need at least $\ceil{\omega(G)/2}$ colors in order to make the coloring triangle-free
    in this clique, and therefore at least $\ceil{\omega(G)/2}$ colors are needed to triangle-free color the entire graph.
  
    The upper bound can be justified by the following argument.
    Let $k=\chi(G)$ and take any classic $k$-coloring of $G$. Let $V_i$ be the set of
    vertices colored $i$, where $1 \leq i \leq k$. For each $0 \leq j \leq \ceil{k/2}-1$ recolor sets $V_{2j+1} \cup V_{2j+2}$
    with $j$ (we may need to add empty set $V_{k+1}$, if $k$ is odd). Since every $V_i$ is an independent set, then after recoloring,
    any monochromatic cycle is of even length. Therefore the resulting coloring is triangle-free.
  \end{proof}

  \begin{remark}
    The recoloring procedure from the above proof we call the {\em standard recoloring strategy}. It will
    be used in the construction of algorithms in the next section.
  \end{remark}
  
  The fundamental question we ask is how tight are the bounds in Theorem \ref{thm:bounds} and what kind of algorithmic
  consequences are implied by these observations. The following
  theorem shows that the upper bound can be arbitrarily large, which means that we should not expect
  to find an algorithm for general graphs that would acceptably approximate $\chi_3(G)$ based on a chromatic number alone.
  
  \begin{theorem}
    For any $k \geq 1$, there exists a graph $G$ for which $\chi_3(G)=1$ and $\chi(G)=k$.
  \end{theorem}

  \begin{proof}
    The class of graphs called Mycielski graphs \cite{mycielski} meet this property, as they are
    triangle-free and can have arbitrarily large chromatic number.
  \end{proof}

  On the other hand, we conjecture that the lower bound of Theorem \ref{thm:bounds} is almost tight.

  \begin{conjecture}\label{con:one}
    For any graph $G$: $\chi_3(G) \leq \ceil{\frac{\omega(G)}{2}}+1$.
  \end{conjecture}

  To show that it is possible to reach the conjectured bound we construct an infinite
  class of graphs for which $\chi_3(G) = \ceil{\omega(G)/2}+1$.
  First, we present an auxiliary class of graphs that we call {\em cycle-cliques}, which will serve as
  building blocks of our construction.
  
  \begin{definition}[cycle-clique]\label{def:kcycleclique}
    Let $k \in \nat$. The {\em k-cycle-clique} is a graph that have exactly $5 \cdot k$ vertices
    $v_{i,j}$ for $i \in \{0, \dots, 4\}$ and $j \in \{0, \dots, k-1\}$
    and there is an edge between $v_{i,j}$ and $v_{i',j'}$ if and only if $|i-_5i'| \leq 1$.
    We will define a set $J_i = \{v_{i,j}: j \in \{0, \dots, k-1\}\}$. Each $J_i$ is called a {\em joint}.
  \end{definition}

  \begin{obs}
    For any $k$-cycle-clique it holds that $J_i \sim K_k$ and $J_i \cup J_{i +_5 1} \sim K_{2k}$, for $i \in \{0, \dots, 4\}$.
  \end{obs}
  
  \begin{example}\label{ex:kcycleclique}
    In Figures 1b 
    and 1c 
    we present $k$-cycle-cliques for $k \in \{1,2\}$.
    In both graphs joints are marked with dashed lines.
    The 1-cycle-clique is simply a cycle of length 5. Notice how $J_i \cup J_{i+_51}$ forms a clique of size $2k$,
    for $i \in \{0,\dots,4\}$, whereas each $J_i$ is a clique of size $k$.
    We also note that the graph in Figure 1c 
    is the same as the gadget
    presented in \cite{karpinski}, where the author uses it to encode what he calls {\em super-edges},
    in his $\mathcal{NP}$-completeness proof of the \textsc{TriangleFree}-2-\textsc{Coloring} problem.
    The cycle-clique is a generalization of this gadget.
  \end{example}
  
  \begin{figure}[t!]
    \captionsetup[subfigure]{labelformat=parens,labelsep=space,font=small,justification=centering}
    \begin{minipage}[c][9cm][t]{.70\linewidth}
      \vspace*{\fill}
      \centering
      \includegraphics[height=7cm]{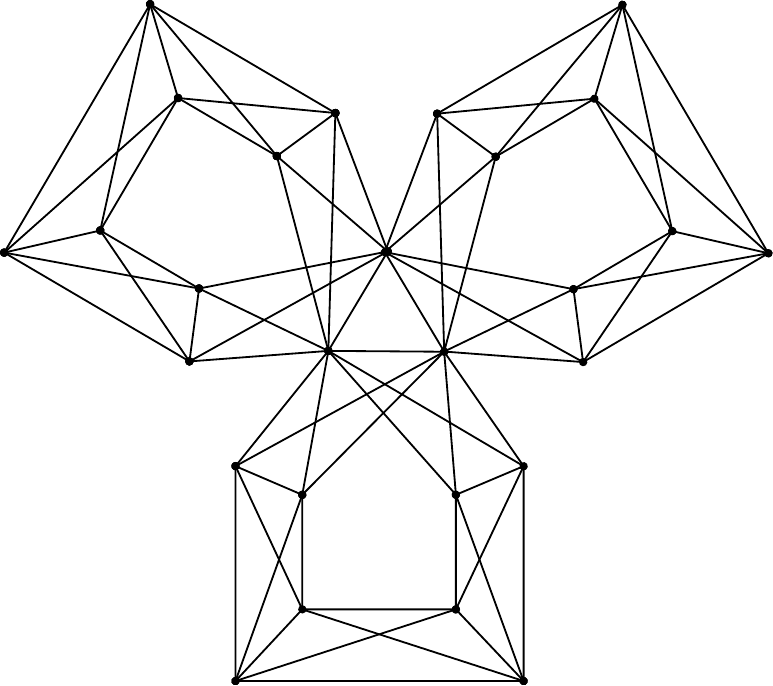}
      \captionof{subfigure}{2-clover graph}\label{fig:tract}
    \end{minipage}
    \begin{minipage}[c][9cm][t]{.25\linewidth}
      \centering
      \def\svgwidth{\linewidth}
      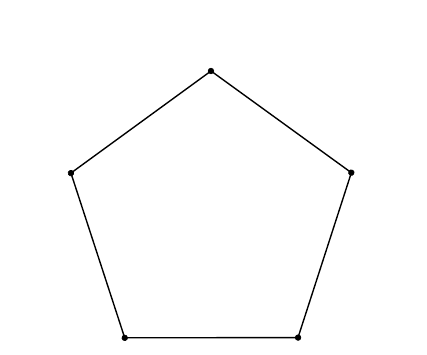
      \captionof{subfigure}{1-cycle-clique}\label{subfig:ex:a}
      \par\vfill
      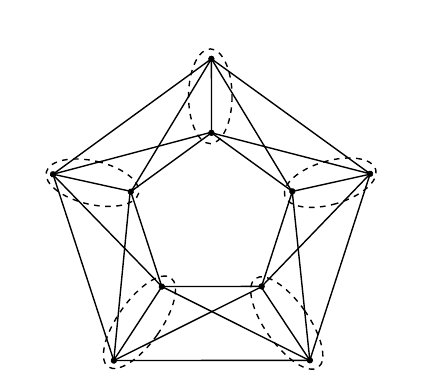
      \captionof{subfigure}{2-cycle-clique}\label{subfig:ex:b}
    \end{minipage}
    \caption{Examples of new classes of graphs.}
    \label{fig:examples}
  \end{figure}

  \begin{obs}\label{obs:2clique}
    Let $k \geq 1$. In every triangle-free k-coloring of graph $K_{2k}$ every color is used exactly twice.
  \end{obs}

  \begin{lemma}\label{lma:diffcolors}
    Let $k \geq 1$. In every triangle-free k-coloring of k-cycle-clique each vertex of $J_i$ has a unique color (w.r.t. other vertices of $J_i$),
    for $i \in \{0, \dots, 4\}$.
  \end{lemma}

  \begin{proof}
    Proof by contradiction.
    Let's say that there is a triangle-free k-coloring of k-cycle-clique and $i \in \{0, \dots, 4\}$
    such that there exist two vertices in $J_i$ that have the same color.
    Using Observation \ref{obs:2clique} we know that this color cannot appear in $J_{i+_51}$ or $J_{i-_51}$.
    Using Observation \ref{obs:2clique} again, this color need to be used twice in each of $J_{i+_52}$ and $J_{i-_52}$.
    Since $J_{i+_52} \cup J_{i-_52} \sim K_{2k}$ from the definition and that four vertices of $J_{i+_52} \cup J_{i-_52}$ are
    colored with the same color, we reach a contradiction with the assumption that the initial coloring was triangle-free.
  \end{proof}

  Now we present the contraction scheme of cliques, which is a procedure used in our construction.

  \vspace{0.2cm}
  
  \noindent {\em Clique contraction scheme:} Let $U=\tuple{u_0,\dots,u_{k-1}}$, $V=\tuple{v_0,\dots,v_{k-1}}$,
  $W=\langle w_0$,$\dots$,$w_{k-1}\rangle$ be the cliques of size $k$. Perform vertex identification on the following
  pairs of vertices (in that order):

  \begin{enumerate}
    \item $(v_i, u_i)$ for $i \in \{0, \dots, k-2\}$,
    \item $(v_i, w_i)$ for $i \in \{0, \dots, k-3\}$,
    \item $(v_{k-1}, w_{k-2})$ and $(u_{k-1}, w_{k-1})$.
  \end{enumerate}

  \begin{example}
    In Figure \ref{fig:ex:scheme} we present a contraction of three cliques: $U=\tuple{u_0,u_1,u_2}$, $V=\tuple{v_0,v_1,v_2}$
    and $W=\tuple{w_0,w_1,w_2}$, using the above procedure. Figure \ref{subfig:exs:a} show the initial state. Dashed lines
    connect vertices that will be identified in the first step of the scheme, namely, $\tuple{v_0,u_0}$ and $\tuple{v_1,u_1}$.
    The result of the first step is presented in Figure \ref{subfig:exs:b}. Next, vertices $\tuple{v_0,w_0}$ are identified
    (Figure \ref{subfig:exs:c}) and lastly, we identify vertices $\tuple{v_2,w_1}$ and $\tuple{u_2,w_2}$. The final result
    is presented in Figure \ref{subfig:exs:d}. Notice that the clique contraction scheme, given three cliques of size $k$, outputs
    one clique of size $k+1$.
  \end{example}

  \begin{figure}[t!]
    \centering
    \subfloat[\label{subfig:exs:a}]{\def\svgwidth{0.22\textwidth}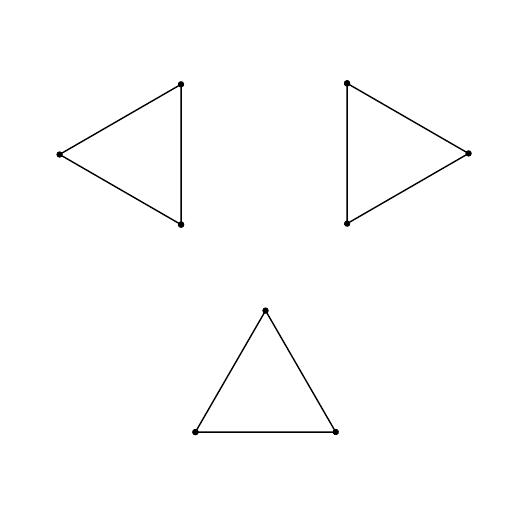}~
    \subfloat[\label{subfig:exs:b}]{\def\svgwidth{0.22\textwidth}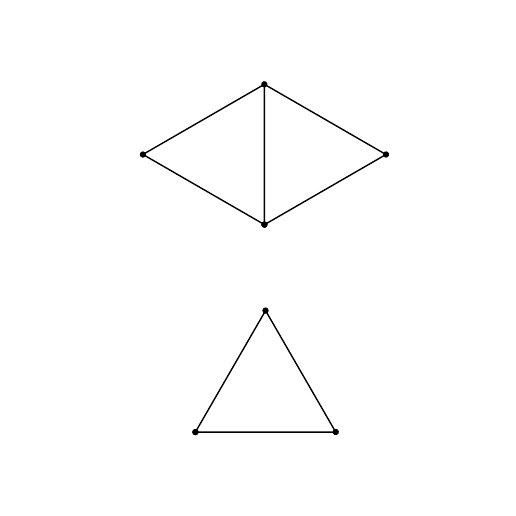}~
    \subfloat[\label{subfig:exs:c}]{\def\svgwidth{0.22\textwidth}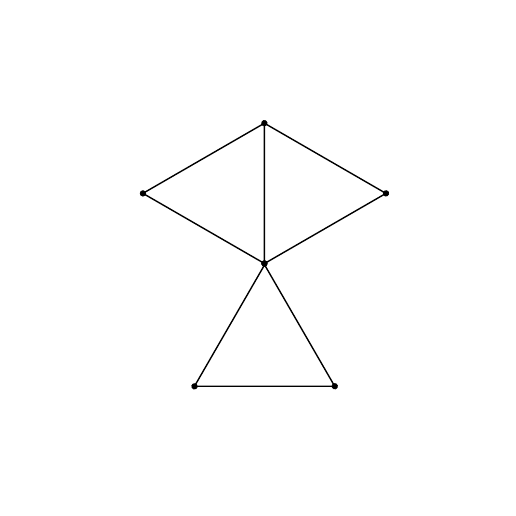}~
    \subfloat[\label{subfig:exs:d}]{\def\svgwidth{0.22\textwidth}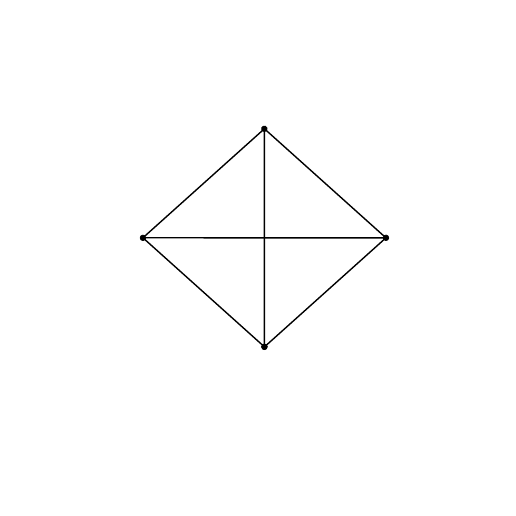}
    \caption{Sample run of the clique contraction scheme.}
    \label{fig:ex:scheme}
  \end{figure}


  We are ready to present the desired construction.
  
  \begin{definition}[k-clover graph]
    Let $G_1=(\cup_{i=0}^4 V_i, E_1)$, $G_2=(\cup_{i=0}^4 U_i, E_2)$ and $G_3=(\cup_{i=0}^4 W_i, E_3)$
    be k-cycle-cliques, where $k \geq 2$.
    A graph created by running the clique contraction scheme on sets $V_0$, $U_0$ and $W_0$
    is called a {\em k-clover graph}.
  \end{definition}

  \begin{example}\label{ex:tract}
    The example of k-clover graph is presented in Figure 1a, 
    for $k=2$. Notice how contracting three joints of size 2 results in a creation of a triangle ($K_3$) in the middle.
  \end{example}
  
  \begin{lemma}\label{lma:2k}
    There is no k-clover graph $G$ such that $K_{2k + 1} \subseteq G$, for $k \geq 2$.
  \end{lemma}

  \begin{proof}
    Only vertices that can have degree at least $2k$ are those that participated in vertex identification.
    In graph $G$ there is exactly $k+1$ such vertices.
  \end{proof}

  \begin{lemma}\label{lma:k+1}
    There is no triangle-free k-coloring of k-clover graph.
  \end{lemma}

  \begin{proof}
    Proof by contradiction.
    By Lemma \ref{lma:diffcolors} we know that in each set $V_0=\tuple{v_0,\dots,v_{k-1}}$,
    $U_0=\tuple{u_0,\dots,u_{k-1}}$ and $W_0=\tuple{w_0,\dots,w_{k-1}}$ vertices need to have different colors.
    Because we identify vertices $(v_i, u_i)$ for $i \in \{0, \dots, k-2\}$, it means that $v_{k-1}$ and $u_{k-1}$
    need to have the same color.
    But because we identify $(v_{k-1}, w_{k-2})$ and $(u_{k-1}, w_{k-1})$ that mean $w_{k-2}$ and $w_{k-1}$ need to have the same color,
    which contradicts Lemma \ref{lma:diffcolors}.
  \end{proof}
  
  \begin{theorem}\label{thm:clover}
    For any $k \geq 2$, there exists a graph $G$, such that $\chi_3(G) = \ceil{\omega(G)/2}+1$.
  \end{theorem}

  \begin{proof}
    Let $G$ be a k-clover graph. From the definition of a cycle-clique, $\omega(G) \geq 2k$ and
    from Lemma \ref{lma:2k} and Theorem \ref{thm:bounds}, $\omega(G) \leq 2k$, therefore $\omega(G)=2k$. We have $\chi_3(G) \geq k + 1$ by Lemma \ref{lma:k+1}.
    We show that there exist a triangle-free coloring of graph $G$ that uses $k+1$ colors. Let $U,V,W$ be the cliques
    used in the contraction scheme, when constructing $G$. The result of the contraction is a clique of size $k+1$. Assign a unique
    color to each of its vertices. Notice that $U$, $V$ and $W$ are cliques of size $k$ where each vertex has a unique color.
    By Lemma \ref{lma:diffcolors} we can extend the coloring of $U$, $V$ and $W$ to their respective k-cycle-cliques. This proves that
    $\chi_3(G) = k + 1 = \frac{2k}{2} + 1 = \ceil{\omega(G)/2}+1$.
  \end{proof}

  \begin{remark}
    Notice that we have used 5 joints in the construction of a single k-cycle-clique. We comment on why constant 5
    is important here. Notice that if we have used even number of joints, then the property of Lemma \ref{lma:diffcolors}
    would not hold. Also, if we were to use 3 joints, then the k-cycle-clique would be isomorphic to $K_{3k}$, and therefore would
    not be triangle-free k-colorable -- we leave the verification of those facts to the reader, as an easy exercise. In conclusion,
    5 is the smallest possible number of joints in the k-cycle-clique for it to hold the desired properties. But in fact, any odd number
    of joints greater or equal to 5 can be used.
  \end{remark}

  Another issue worth discussing is that the proof of Theorem \ref{thm:clover}
  shows the existence of graphs for which $\chi_3(G) = \ceil{\omega(G)/2}+1$, but only
  when $\omega(G)$ is even. We now show that there exists a construction for $r=3$, and
  we leave the existence of such graphs for other values of $r$ as an open problem.

  \begin{definition}\label{def:k4free}
    The {\em $K_4$-free polar gadget} is the graph $G=(V,E)$ with $V=\{u$, $v$, $w_1$, $w_2$, $w_3$, $w_4$, $z_1$, $z_2$, $z_3$, $y_1$, $y_2$, $y_3\}$, and edge set
    defined as in Figure \ref{fig:k4free}.
  \end{definition}

  \setlength{\abovecaptionskip}{2pt}
  \setlength{\belowcaptionskip}{-10pt}

  \begin{figure}[t!]
    \def\svgwidth{0.5\linewidth}\centering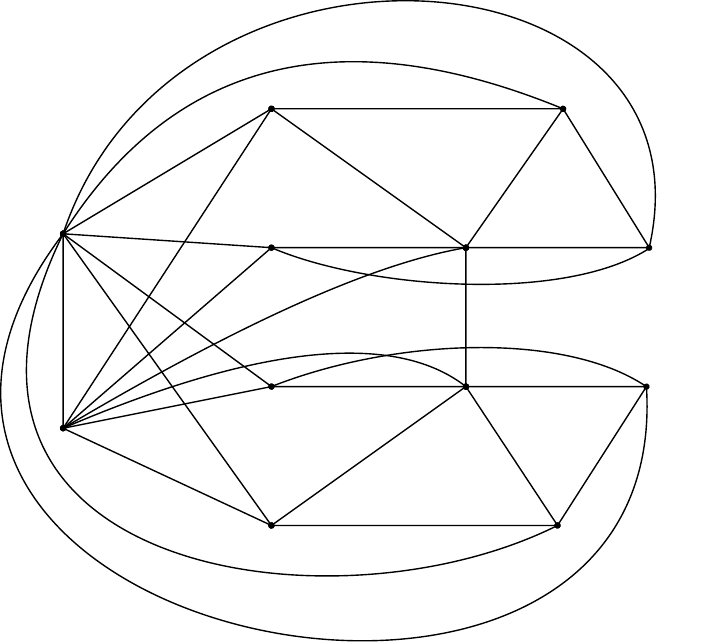
    \caption{The $K_4$-free polar gadget. It consists of 12 vertices, 30 edges and 20 triangles.}
    \label{fig:k4free}
  \end{figure}
  
  \noindent We say that egde $xy$ is {\em polar}, if in every triangle-free 2-coloring $c$, $c(x) \neq c(y)$.
  
  \begin{lemma}\label{lma:k4free}
    Let $G$ be a $K_4$-free polar gadget. Then the following are true:
    (i) $G$ is triangle-free 2-colorable,
    (ii) $uv$ is a polar edge,
    (iii) $G$ is $K_4$-free.
  \end{lemma}

  \begin{proof}
    {\em Ad. (i):} Let $S=\{u,z_1,y_1\}$. Set $c(S)=1$ and $c(V \setminus S)=2$.

    {\em Ad. (ii):} Assume by contradiction, that $c(u)=c(v)=1$ and $G$ is triangle-free 2-colorable. Then $c(\{w_1,w_2,w_3,w_4\})=2$, which is forced.
    Since $c(v)=1$, then either $c(z_1)=2$, or $c(y_1)=2$. Assume the former, without loss of generality. Then $c(z_2)=1$, otherwise
    $\{w_1,z_1,z_2\}$ would be monochromatic. But this means that we cannot color $z_3$, because if $c(z_3)=1$, then $\{u,z_2,z_3\}$ is monochromatic, and
    if $c(z_3)=2$, then $\{w_2,z_1,z_3\}$ is monochromatic -- a contradiction.

    {\em Ad. (iii):} Let $V_1=\{u,v\}$, $V_2=\{w_1,w_2,w_3,w_4\}$, $V_3=\{z_1,z_2,z_3,y_1,y_2,y_3\}$.
    If there exists $H \subset G$ that is isomorphic to $K_4$, then we can observe that:

    \begin{itemize}
      \item $H$ contains at most one vertex from $V_2$, since $V_2$ is an independent set, and
      \item $H$ contains at most one vertex from $V_1$, because otherwise $H$ would contain
        at least one vertex from $V_3$, and there is no vertex (or pair of vertices) in $V_3$ that is connected to both
        vertices in $V_1$.
    \end{itemize}

    \noindent Therefore we are left with the following four cases:

    \begin{itemize}
      \item $H$ uses only vertices from $V_3$. But in $G(V_3)$ only $z_1$ and $y_1$ have degree at most $3$.
      \item $H$ uses $3$ vertices from $V_3$ and one vertex from $V_2$. But each vertex from $V_2$ is adjacent to exactly
        two vertices from $V_3$.
      \item $H$ uses $3$ vertices from $V_3$ and one vertex from $V_1$. Then $v \not\in H$
        because it is adjacent to only $z_1$ and $y_1$. Thus $u \in H$, but it easy to verify that
        $N(u)$ does not contain a triangle.
      \item $H$ uses $2$ vertices from $V_3$, one vertex from $V_2$ and one vertex from $V_1$. But for
        each $\{s,t\} \subset V_3$, there exists exactly one triangle $\{s,t,r\}$, such that $r \in V_1 \cup V_2$.
    \end{itemize}

    \noindent We have reached contradictions in all possible cases, thus we conclude that $G$ is $K_4$-free.
  \end{proof}

  \begin{theorem}\label{thm:k4free}
    There exists a graph $G$, such that $\chi_3(G) = \ceil{\omega(G)/2}+1$, and $\omega(G)=3$.
  \end{theorem}

  \begin{proof}
    To construct $G$, take three $K_4$-free polar gadgets and identify their polar edges so that they form a triangle. Properties (i) and (ii)
    of Lemma \ref{lma:k4free} forces the use of the third color in $G$ and Property (iii) makes the construction $K_4$-free.
  \end{proof}

  We believe that the gadgets that we presented in this section are interesting from the graph-theoretic point of view.
  It would be instructive to find answers to the following questions.

  \begin{question}
    For any $k \geq 2$, what is the smallest graph (in terms of vertices and/or edges) for which the following are true:
    (i) $\chi_3(G)=k$, and (ii) for any triangle-free k-coloring of $G$, there exists a {\em polar clique} of size k in $G$, i.e., a clique
    of size k where each vertex has a unique color?
  \end{question}

  \begin{question}
    Is $K_4$-free polar gadget the smallest is terms of vertices/edges/triangles graph that satisfy the properties of Lemma \ref{lma:k4free}?
  \end{question}
  
  To finish this section we give another bound on $\chi_3(G)$. We recall the well-known Brook's Theorem \cite{brooks} which
  states that for any graph $G$, we have $\chi(G) \leq \Delta(G)$,
  unless $G$ is a complete graph or an odd cycle. Using this theorem we can prove the following.

  \begin{theorem}\label{thm:delta}
    Let $G=(V,E)$ be any graph where $|V|>3$, where $G$ is not a complete graph of odd number of vertices.
    Then $\chi_3(G) \leq \ceil{\frac{\Delta(G)}{2}}$.
  \end{theorem}

  \begin{proof}
    If $G$ is an odd cycle of length at least 5, then $\chi_3(G)=1$ and $\ceil{\frac{\Delta}{2}}=1$, so
    the inequality holds. If $G$ is a complete graph (a clique) of $n$ vertices, then $\chi_3(G)=\ceil{\frac{n}{2}}$
    and $\Delta=n-1$. The inequality $\ceil{\frac{n}{2}} \leq \ceil{\frac{n-1}{2}}$ holds iff $n$ is even.
    If the above cases does not occur, then the application of the Brook's Theorem
    combined with the upper bound of Theorem \ref{thm:bounds} completes the proof.
  \end{proof}
  
  \section{Tractable classes of graphs}

  Here we present several classes of graphs for which efficient algorithms
  for triangle-free coloring problem exist. We use the bounds derived in the previous section
  to show that for small values of structural parameters polynomial time complexity
  can be achieved in many cases.

  Given graph $G=(V,E)$, we distinguish between two optimization variants of the triangle-free coloring problem:
  (i) finding the number $\chi_3(G)$, and (ii) finding the triangle-free coloring that uses exactly
  $\chi_3(G)$ colors, i.e., the mapping from $V$ to $\{1,\dots,\chi_3(G)\}$ is the required output.
  As we will soon see, complexities for these two variants can be different.
  To this end we propose the following notion of time complexity for our problem.
  For convenience, we say that $n$ is the number of vertices of graph $G$ and $m$ is the number of edges of $G$,
  unless stated otherwise.

  \begin{definition}
    Let $\mathcal{G}$ be a class of graphs. We say that the triangle-free coloring
    problem is solvable in time $O(f(n+m),g(n+m))$ on $\mathcal{G}$, iff there exist
    algorithms $\mathcal{A}$ and $\mathcal{B}$, such that for every input $G \in \mathcal{G}$,
    (i) algorithm $\mathcal{A}$ outputs $\chi_3(G)$ in $O(f(n+m))$ time, and (ii) algorithm $\mathcal{B}$
    outputs the triangle-free coloring that uses exactly $\chi_3(G)$ colors, in $O(g(n+m))$ time.
  \end{definition}

  \subsection{Graphs with bounded chromatic number}

  Here we present a single theorem that presents the complexity of triangle-free coloring problem on
  several popular classes of graphs. We exploit the fact that if we know in advance that $\chi(G) \leq 4$ --
  and therefore $\chi_3(G) \leq 2$, by Theorem \ref{thm:bounds} -- then $\chi_3(G)$ can be found easily and is one of the following:
  (a) $\chi_3(G)=0$ iff $G$ is empty, (b) $\chi_3(G)=1$ iff $G$ is triangle-free ($K_3$-free), (c) $\chi_3(G)=2$ iff the above two cases does not hold.
  Checking the first case is trivial. Therefore the challenge of finding $\chi_3(G)$ lies in checking whether $G$ contains a triangle.

  For finding the actual coloring, we use the fact that in the given graphs we can find ``good-enough''
  classic coloring, and then recolor vertices so that there are no monochromatic triangles using the {\em standard recoloring strategy}.

  \begin{theorem}\label{thm:positive}
    The triangle-free coloring problem is solvable:

    \begin{itemize}
      \item in time $O(n,n^2)$ on planar graphs,
      \item in time $O(n,n)$ on: outerplanar graphs, chordal graphs, $\Delta$-regular graphs, with $\Delta \leq 4$.
    \end{itemize}
  \end{theorem}
  
  We first give the required definitions, then, the proof.
    
  We say that the graph is {\em planar}, if we can embed it in the plane,
  i.e., it can be drawn on the plane in such a way that its edges intersect
  only at their endpoints. Graph $G$ is {\em outerplanar} if it is planar and there exists a plane
  embedding of $G$ in which all vertices lie on the exterior (unbounded) face.
  Graph is {\em chordal} if every cycle of length greater than 3 has a chord -- an edge
  connecting two non-adjacent vertices on a cycle.
  Graph is {\em $\Delta$-regular}, if its every vertex has degree $\Delta$.

  \begin{proof}[Proof of Theorem \ref{thm:positive}]

    We begin with chordal graphs. Chordal graphs are the part of the larger class called perfect graphs.
    A perfect graph $G$ has a property $\chi(G)=\omega(G)$, therefore
    $\chi_3(G)=\ceil{\frac{\chi(G)}{2}}=\ceil{\frac{\omega(G)}{2}}$, by Theorem \ref{thm:bounds}. This means that
    determining $\chi_3(G)$ is equivalent to determining either $\omega(G)$ or $\chi(G)$.
    Fortunately, this can be done in linear-time in chordal graphs \cite{golumbic}. Moreover, not only can
    we find $\chi(G)$ in linear-time, but also the (classic) coloring that uses $\chi(G)$
    colors. Using standard recoloring strategy, we can produce a triangle free coloring that uses
    $\ceil{\frac{\chi(G)}{2}}=\chi_3(G)$ colors.

    We now turn out attention to other classes of graphs. First, we show that we can find whether
    given graph $G$ is triangle-free in $O(n)$ time, for the following cases:

    \begin{itemize}
      \item {\em planar graphs} - using linear-time algorithm by Papadimitriou and Yannakakis \cite{papadimitriou},
      \item {\em outerplanar graphs} - same as above,
      \item {\em $\Delta$-regular graphs} - we loop through all the vertices, and for every vertex $v$,
        check if any pair of vertices from $N(v)$ are adjacent to each other. This takes at most $O(n \cdot \Delta^2)=O(n)$ time.
    \end{itemize}

    \noindent Hence, we can check if $\chi_3(G) > 1$ in linear-time (checking if $\chi_3(G) = 0$ is trivial).
    Furthermore, from Theorems \ref{thm:bounds} and \ref{thm:delta} we know that $\chi_3(G) \leq 2$ in all three cases, therefore
    determining $\chi_3(G)$ can be done in time $O(n)$.

    From now on we assume that the input graph contains at least one triangle.
    Otherwise, we can output a trivial coloring $c(V)=1$.

    Finding the actual coloring of $G$ has different time complexity depending
    on the class of $G$. First, we find a ``good-enough'' classic coloring of $G$:

    \begin{itemize}
      \item {\em planar graphs} - The well-known Four Color Theorem states
        that every planar graph is classically 4-colorable, that is, for each planar $G$, $\chi(G) \leq 4$.
        The theorem was proved by Appel and Haken in \cite{appel}, where the authors propose a quadratic-time
        algorithm based on their proof. The algorithm was later improved by Robertson et al. \cite{4coloring},
        but the asymptotic running time is still $O(n^2)$.
      \item {\em outerplanar graphs} -  Proskurowski and Sys{\l}o \cite{syslo}
        give an algorithm for finding $\chi(G)$ on outerplanar graphs which runs in
        $O(n)$ time.
      \item {\em $\Delta$-regular graphs (where $\Delta \leq 4$)} - Skulrattanakulchai \cite{delta-list} shows
        how to find classic $\Delta$-coloring in $O(n)$ time.
    \end{itemize}

    Given the classic coloring of $G$, we use the standard recoloring strategy to obtain the final coloring.
    Since we know that $\chi(G) \leq 4$ and $\chi(G) \geq 3$ (by the assumption that is contains at least one triangle),
    then $\chi_3(G) = 2$ and the resulting coloring is optimal.
  \end{proof}

  We do not know if there exists a sub-quadratic algorithm for finding a triangle-free
  coloring in planar graphs, so we leave it as an open problem. It is worth noting
  that the situation would not improve even if we knew that the the input graph is classically
  3-colorable, as the best know algorithm for classically 4-coloring planar graphs with this property
  still runs in $O(n^2)$ time \cite{kawarabayashi}.

  \begin{question}
    Is there an algorithm that finds a triangle-free coloring of a planar graph in time $o(n^2)$?
  \end{question}

  \begin{remark}
    Theorem \ref{thm:positive} serves as an evidence, that the triangle-free coloring problem is
    easier than the classic coloring problem. We have efficient algorithms for finding $\chi_3$ on
    planar graphs and graphs with $\Delta=4$, but it is $\mathcal{NP}$-hard to find the $\chi$ even on
    4-regular planar graphs \cite{4regplanar}. This suggests that the above question might have
    a positive answer.
  \end{remark}
  
  \subsection{FPT result}

  In this subsection we look at our problem from the parametrized complexity point of view.
  The main goal here is not to give efficient algorithms, but rather provide the initial
  classification of the problem in the context of $\mathcal{FPT}$ theory. We prove that
  \textsc{TriangleFree}-q-\textsc{Coloring} problem is fixed-parameter tractable when
  parametrized by vertex cover number.
  
  In the article by Fiala, Golovach and Kratochv\'il \cite{fiala}, authors explore differences
  in the complexity of several coloring problems when parametriazed by either vertex cover number or treewidth.
  We use techniques similar to theirs in the proof of the following theorem.
  
  \begin{theorem}\label{thm:vc}
    The \textsc{TriangleFree-}q\textsc{-Coloring} problem is $\mathcal{FPT}$ when
    parametrized by the vertex cover number.
  \end{theorem}

  \begin{proof}
    Let $W$ be a minimum vertex cover of $G$, and let $|W|=k$. Then $I = V \setminus W$ is an independent set. The
    goal is to find a triangle-free coloring $c : V \rightarrow \{1,\dots,q\}$.

    The algorithm works in two steps: first, we find the triangle-free q-coloring of $W$
    by exhaustive search. Then, since $I$ is an independent set, we can use greedy strategy
    to color its vertices, once we know the coloring of $W$. By greedy strategy we mean taking
    each $v \in I$, checking its colored neighborhood $N(v)$ (notice that $N(v) \subseteq W$), and coloring
    $v$ in a way which does not create a monochromatic triangle in $G$. Clearly this procedure will find
    a triangle-free q-coloring iff such coloring exists. The analysis of running time now follows.

    Since $\chi_3(W) \leq \ceil{\chi(W)/2} \leq \ceil{k/2}$, it is natural to consider the following two cases:

    (i) If $q \leq \ceil{k/2}$, then we consider all q-colorings of $W$ and their extensions to $I$
    by the greedy algorithm. Since $W$ has at most $k^{\ceil{k/2}}$ colorings, the running time is
    $O(k^{\ceil{k/2} + 1}n)$.

    (ii) If $q > \ceil{k/2}$, then assuming that $W=\{w_0,\dots,w_{k-1}\}$, we can set $c(w_{2i})=c(w_{2i+1})=i+1$,
    for each $0 \leq i < \ceil{k/2}$, which clearly uses $\ceil{k/2}$ colors, then set $c(I)=\ceil{k/2}+1$. Thus,
    in this case, the triangle-free q-coloring always exists, and we can find it in $O(n)$ time.

    We conclude that the total running time of the algorithm is $O(k^{\ceil{k/2} + 1}n)$, which proves the claim.
  \end{proof}
  
  \section{Hardness results}

  Here we improve the results of Karpi\'nski \cite{karpinski} and Shitov \cite{shitov}. In their papers they show that the
  generalized version of \textsc{TriangleFree}-2-\textsc{Coloring} problem is $\mathcal{NP}$-complete -- they not only consider monochromatic triangles,
  but monochromatic cycles of arbitrary (fixed) length. Karpi\'nski first presented the proof, which was then simplified by Shitov, by introducing
  an auxiliary problem, which we call the \textsc{TriangleFreePolar}-2-\textsc{Coloring} problem. In Shitov's proof, this new problem
  serves as a connection between \textsc{NotAllEqual 3-SAT} and \textsc{TriangleFree}-2-\textsc{Coloring}.

  Our contribution is as follows: we first prove that the triangle-free coloring remains
  $\mathcal{NP}$-hard even if we can use any fixed number of colors. Then we show that the 
  \textsc{TriangleFree}-2-\textsc{Coloring} problem is $\mathcal{NP}$-hard in the restricted class
  of graphs, which do not contain $K_4$ as a subgraph. These improve the results of Karpi\'nski \cite{karpinski}.
  Lastly, we show that the \textsc{TriangleFreePolar}-2-\textsc{Coloring} remains $\mathcal{NP}$-hard even in graphs
  of maximum degree 3, which improves the result of Shitov \cite{shitov}.
  
  \begin{theorem}\label{thm:hard}
    For any fixed $q \geq 2$, the \textsc{TriangleFree}-q-\textsc{Coloring} problem is $\mathcal{NP}$-hard.
  \end{theorem}

  \begin{proof}
    We prove the theorem by induction on $q$. For $q=2$ the hardness has already been proved by Karpi\'nski \cite{karpinski}.
    What remains is to show the polynomial reduction from the \textsc{TriangleFree}-q-\textsc{Coloring} problem to
    \textsc{TriangleFree}-(q+1)-\textsc{Coloring} problem, for any fixed $q \geq 2$. Take any graph $G=(V,E)$,
    where $|V|=n$ and $|E|=m$, and assume that the vertices are arbitrarily ordered, i.e., $V=\{v_1,\dots,v_n\}$.
    We construct the graph $G'=(V',E')$ in the following way:

    \begin{enumerate}
    \item Let $V' = V \cup \left(\bigcup_{i=1}^n V_i\right)$ and $E' = E \cup \left(\bigcup_{i=1}^n E_i\right)$,
      where $V_i = V^0_i \cup V^1_i \cup V^2_i \cup V^3_i \cup V^4_i$ and $(V_i,E_i)$
      is a $(q+1)$-cycle-clique with the set $\{V^0_i, V^1_i, V^2_i, V^3_i, V^4_i\}$ being its joints.
    \item For every $1 \leq i \leq n$, pick arbitrary vertex from $V^0_i$ and call it $u_i$. Identify the pairs of vertices in order:
      $\tuple{u_1,u_2},\dots,\tuple{u_1,u_n}$. Rename $u_1$ to $u$.
    \item For every $1 \leq i \leq n$, take any $s \in V^0_i$ such that $s \neq u$ and identify vertices $\tuple{v_i,s}$.
    \end{enumerate}

    Observe that $v_iu$ is an edge in $V^0_i$ (for $1 \leq i \leq n$) and therefore in any triangle-free $(q+1)$-coloring of $G'$, $v_i$
    has a different color than $u$, by Lemma \ref{lma:diffcolors}. It is now easy to verify that $G$ is triangle-free $q$-colorable iff
    $G'$ is triangle-free $(q+1)$-colorable. Also, since $q$ is part of the problem (and therefore a constant), the reduction is polynomial
    w.r.t. $n$ and $m$.
  \end{proof}

  \begin{theorem}\label{thm:hardk4}
    The \textsc{TriangleFree}-2-\textsc{Coloring} problem is $\mathcal{NP}$-hard on $K_4$-free graphs.
  \end{theorem}

  \begin{proof}
    We create a polynomial reduction from the following problem, which is known to be $\mathcal{NP}$-complete \cite{naesat}:

    \vspace{\baselineskip}
    
    \noindent \textsc{NotAllEqual 3-Sat}

    \noindent {\em Input:} Boolean formula $\phi$ in conjunctive normal form, where each clause consists of exactly
    three literals.

    \noindent {\em Question:} Does $\phi$ have a nae-satisfying assignment, i.e., in each clause at least one literal is true
    and at least one literal is false?

    \vspace{\baselineskip}

    Let $\phi$ be a Boolean formula as described above, with $n$ variables $X$ and $m$ clauses $C_1,\dots,C_m$.
    We assume, without loss of generality, that each clause consists of at least two unique literals, otherwise
    the instance is trivially false. We show a construction of $K_4$-free graph $G=(V,E)$,
    such that $\phi$ is nea-satisfiable iff there exist a triangle-free 2-coloring of $G$.
    
    \begin{enumerate}
    \item For every clause $C_i \equiv (l^i_1 \vee l^i_2 \vee l^i_3)$, $1 \leq i \leq m$,
      add vertices $l^i_1$, $l^i_2$, and $l^i_3$ to $V$, and add edges $\{l^i_1,l^i_2\}$, $\{l^i_2,l^i_3\}$ and $\{l^i_3,l^i_1\}$ to $E$.
    \item For every variable $x \in X$, add vertices $t_x$ and $f_x$ to $V$, and edge $\{t_x, f_x\}$ to $E$.
    \item For every variable $x$ and every occurrence of literal $l \equiv \neg x$ in $\phi$, add edge $\{t_x,l\}$ to $E$.
      For every variable $x$ and every occurrence of literal $l' \equiv x$ in $\phi$, add edge $\{f_x,l'\}$ to $E$.
    \item For every edge $e \in E$, except for edges added in step 1, create $K_4$-free polar gadget (Definition \ref{def:k4free}),
      with $e$ begin the polar edge.
    \end{enumerate}

    Observe that introducing polar gadgets in step 4 forces every vertex corresponding to the literal $l$ to have to same color, and
    also every vertex corresponding to literal $\neg l$ to have the same color, but different than the color of $l$'s. Therefore the nae-satisfying assignment
    of each clause can be obtained from the coloring of each triangle introduced in step 1. Note that $K_4$-free polar gadget is 2-colorable, by Lemma \ref{lma:k4free}.
    From the above observations, we can conclude that $\phi$ has a nae-satisfying assignment iff $G$ is triangle-free 2-colorable.
    Also, the reduction is polynomial w.r.t. $n$ and $m$, and $G$ is $K_4$-free, since every polar gadget used is $K_4$-free, by Lemma \ref{lma:k4free}.
  \end{proof}

  \begin{theorem}\label{thm:polar}
    The \textsc{TriangleFreePolar}-2-\textsc{Coloring} problem is $\mathcal{NP}$-hard on graphs with maximum degree at most 3, but it
    is linear-time solvable on graphs with maximum degree at most 2.
  \end{theorem}

  \begin{proof}
    For hardness part, we create a polynomial reduction from the following problem,
    which we have proved to be $\mathcal{NP}$-complete (see Appendix \ref{apx:nae}):

    \vspace{0.3\baselineskip}
    
    \noindent \textsc{NotAllEqual 3-Sat-4}

    \noindent {\em Input:} Boolean formula $\phi$ in conjunctive normal form, where each clause consists of exactly
    three literals and every variable appears at most four times in the formula.

    \noindent {\em Question:} Does $\phi$ have a nae-satisfying assignment, i.e., in each clause at least one literal is true
    and at least one literal is false?

    \vspace{0.3\baselineskip}

    Let $\phi$ be a Boolean formula as described above, with $n$ variables $X$ and $m$ clauses $C_1,\dots,C_m$.
    We show a construction of graph $G=(V,E)$ and $S \subseteq E$, such
    that the maximum vertex degree in $G$ is at most 3, and $\phi$ is nae-satisfiable iff there exist a
    triangle-free 2-coloring of $G$ with the set $S$ of polar edges.

    \begin{enumerate}
      \item For every clause $C_i \equiv (l^i_1 \vee l^i_2 \vee l^i_3)$, $1 \leq i \leq m$,
        add vertices $l^i_1$, $l^i_2$, and $l^i_3$ to $V$, and add edges $\{l^i_1,l^i_2\}$, $\{l^i_2,l^i_3\}$ and $\{l^i_3,l^i_1\}$ to $E$.
      \item For every variable $x \in X$ and $\alpha \in \{t,f\}$, let $B^{\alpha}_x = (V^{\alpha}_x, E^{\alpha}_x)$, where
        $V^{\alpha}_x = \{\alpha^1_x,\dots, \alpha^7_x\}$ and
        $E^{\alpha}_x = \left\{ \{\alpha^i_x, \alpha^{2i}_x\}, \{\alpha^i_x, \alpha^{2i+1}_x\} \, \middle| \, i \in \{1,2,3\} \right\}$, i.e.,
        $B^{\alpha}_x$ is a complete binary tree of height 2, with $\alpha^1_x$ being the root, and $\{\alpha^4_x,\alpha^5_x, \alpha^6_x,\alpha^7_x \}$
        the set of leafs. For each $x \in X$ add $V^{t}_x \cup V^{f}_x$ to $V$ and $S_1 = E^{t}_x \cup E^{f}_x \cup \{t^1_x,f^1_x\}$ to $E$.
      \item For every $x \in X$, do the following: take every occurrence of literal $x$ (let's say it occurs $p \leq 4$ times)
        and order their corresponding vertices in $V$ in an arbitrary way: $x^t_1,\dots,x^t_p$. Do similar with literals $\neg x$,
        and let the resulting sequence be $x^f_1,\dots,x^f_q$, for some $q \leq 4$.
        Add edges $S_2 = \left\{ \{x^t_i,t^{i+3}_x\} \, \middle| \, 1 \leq i \leq p \right\} \cup \left\{ \{x^f_i,f^{i+3}_x\} \, \middle| \, 1 \leq i \leq q \right\}$ to $E$.
        Here, we have connected all occurrences of positive literals $x$ to the distinct leafs of $B^t_x$, and all occurrences of negative literals $\neg x$ to the
        distinct leafs of $B^f_x$.
      \item Let $S=S_1 \cup S_2$.
    \end{enumerate}

    It is easy to verify that every vertex in $G$ has degree at most 3, and that we can construct $G$ in time polynomial w.r.t. $n$ and $m$.
    The correctness of the reduction now follows.

    Assume that $\phi$ has a satisfying assignment $A : X \rightarrow \{true, false\}$,
    such that in each clause at least one literal is true
    and at least one literal is false. For every variable $x$ and every
    vertex $x^t_i$ and $x^f_i$, set $c(x^t_i)=1$ and $c(x^f_i)=2$, if $A(x) = true$,
    and $c(x^t_i)=2$ and $c(x^f_i)=1$, if $A(x) = false$. With this partial coloring we
    eliminated any possibility of there being any monochromatic triangle in $G$. We extend
    this coloring to the rest of the graph. Notice that by the choice of $S$, we have ensured
    that in every coloring of $B^{\alpha}_x$,
    $c(\alpha^1_x) \neq c(\{\alpha^2_x, \alpha^3_x\}) \neq c(\{\alpha^4_x, \alpha^5_x, \alpha^6_x, \alpha^7_x\})$.
    Furthermore, since we have introduced the edge $\{t^1_x,f^1_x\}$, the coloring of $B^{t}_x$ is simply a reflection
    of the coloring of $B^{f}_x$, in particular,
    $c(\{t^4_x, t^5_x, t^6_x, t^7_x\}) \neq c(\{f^4_x, f^5_x, f^6_x, f^7_x\})$. Therefore to properly color the vertices
    of $G$ (to satisfy the choice of $S$), we set $c(t^1_x)=2$ if $A(x)=true$ and $c(t^1_x)=1$ if $A(x)=false$. The
    coloring can then be extended to the rest of $B^t_x$ and to $B^f_x$ in a straightforward way.

    Now assume that there exists a triangle-free 2-coloring of $G$. Then by the coloring property of $B^t_x$ and $B^f_x$ mentioned
    earlier, every occurrence of literal $x$ will have the same color, let's say 1, and every occurrence of $\neg x$ will be colored
    2. Since each triangle corresponding to the clause $C_i \equiv (l^i_1 \vee l^i_2 \vee l^i_3)$ uses two colors, we can build the
    satisfying assignment $A$ of $\phi$, such that in each clause at least one literal is true
    and at least one literal is false, by setting $A(l)=true$ if $c(l)=1$ and $A(l)=false$ if $c(l)=2$.

    When $\Delta(G) \leq 2$, then $G$ is a union of isolated vertices, paths and/or cycles. Therefore $G$ is not triangle-free 2-colorable
    iff $G$ contains an odd cycle in which every edge is polar. The existence of such cycle can be checked in linear time.
    
  \end{proof}
  
  \section{Conclusions and Future Work}  

  In this paper we have shown several results regarding vertex coloring without monochromatic triangles.
  Many new interesting problems can be derived from this new coloring variant.
  We have already asked a few questions throughout this paper,
  in particular, it would be instructive to see the answer to Conjecture \ref{con:one}.
  Apart from that, we propose a handful of ways to extend our research:

  \begin{itemize}
    \item For positive results one can look for an algorithm that finds $\chi_3(G)$ in graphs with
      $\Delta(G) \geq 5$. Using an algorithm from \cite{delta-list} when $\Delta(G) = 5$
      can only get us classic 5-coloring of $G$, and therefore using standard recoloring strategy
      may not produce an optimal solution. This issue requires more complicated algorithmic approach.
    \item For negative side, one can look for the smallest $\Delta(G)$, for which the
      \textsc{TriangleFree}-q-\textsc{Coloring} problem is $\mathcal{NP}$-hard.
    \item In the context of parametrized complexity we have only shown two results. Constructing algorithms for
      other choices of parameters is a fine research direction. It would also be good to know where the problems proved to
      be $\mathcal{NP}$-hard in Section 5 reside in $W$-hierarchy. From the proofs of the
      theorems in Section 5 it seems that one will need to first
      consider parametrized complexity of the \textsc{NotAllEqual SAT} problem (and its variants)
      to achieve meaningful results.
  \end{itemize}

  Finally, we note that it is possible to further generalize the notion of $\chi_3$ to get
  the parameter $\chi_r$, for any $r \geq 3$ (notice that $\chi_2=\chi$). This new parameter
  can restrict the coloring in two ways: either not allow monochromatic $K_r$, or not allow
  monochromatic $C_{r'}$ (cycle of length $r'$), for any $3 \leq r' \leq r$. Both extensions are interesting.

  \newpage
  
  \begin{appendices}
    \section{NotAllEqual 3-SAT-4 is NP-complete}\label{apx:nae}

    For the sake of the proof below we define the following: a clause $C$ is NAE-satisfied by assignment $A$, if there exists $l_1 \in C$, such that
    $A(l_1)=true$ and there exists $l_2 \in C$, such that $A(l_2)=false$. A CNF formula $\psi$ is NAE-satisfied by $A$, if all its clauses
    are NAE-satisfied by $A$.
    
    \begin{proof}
      We construct a polynomial reduction from the following $\mathcal{NP}$-complete problem:

      \vspace{\baselineskip}
      
      \noindent \textsc{3-Sat-4}

      \noindent {\em Input:} Boolean formula $\phi$ in conjunctive normal form, where each clause consists of exactly
      three literals and every variable appears exactly four times in the formula.

      \noindent {\em Question:} Does $\phi$ have a satisfying assignment?

      \vspace{\baselineskip}

      \noindent The \textsc{3-Sat-4} problem was proved to be $\mathcal{NP}$-complete by Tovey \cite{tovey}.

      Let $\phi$ be the Boolean formula in CNF, as described above. Let $X$ be the set of variables in $\phi$ and
      let $|X|=n$. Let $C = \{C_1,\dots,C_m\}$ be the set of clauses of $\phi$. 
      We now construct $\phi'$, the instance of \textsc{NotAllEqual 3-Sat-4}.

      The variable set of $\phi'$ is $X \cup \{c_1,\dots,c_m\} \cup \{f_1,\dots,f_m\}$. We add the following clauses
      to $\phi'$:

      \begin{enumerate}
        \item $(x \vee y \vee c_i) \wedge (z \vee \neg c_i \vee f_i)$, for every clause $C_i \equiv (x \vee y \vee z)$, where $1 \leq i \leq m$,
        \item $(\neg f_m \vee \neg f_m \vee f_{1}) \wedge \bigwedge_{i=1}^{m-1} (\neg f_i \vee \neg f_i \vee f_{i+1})$
      \end{enumerate}

      First, observe that in every assignment $A'$ of $\phi'$ and every $1 \leq i,j \leq m$, $A'(f_i)=A'(f_{j})$.
      This property is maintained by clauses from 2. This also implies that clauses from 2. are always NAE-satisfied.

      Assume that $A$ is an assignment satisfying $\phi$. The assignment that NAE-satisfies $\phi'$ is then defined as follows:

      \begin{itemize}
        \item For every $x \in X$, set $A'(x)=A(x)$,
        \item For every $C_i \equiv (x \vee y \vee z)$, where $1 \leq i \leq m$, if $A(x)=A(y)=false$, then set $A'(c_i)=true$, otherwise $A'(c_i)=false$.
        \item For every $1 \leq i \leq m$, set $A'(f_i)=false$.
      \end{itemize}

      \noindent Clearly, assignment $A'$ NAE-satisfies $\phi'$.

      Now assume that $A'$ is some NAE-satisfying assignment of $\phi'$. If in this assignment all $f_i$'s are true,
      then we can negate the assignment of all the variables of $\phi'$, which would still result in a NEA-satisfying assignment.
      Therefore we can assume without loss of generality, that all $f_i$'s are false in $A'$. Take any $C_i \equiv (x \vee y \vee z)$, for 
      $1 \leq i \leq m$. Because of clauses introduced in 1., if $A'(c_i)=true$, then $A'(z)=true$, and if $A'(c_i)=false$, then $A'(x)=true$ or $A'(y)=true$.
      Therefore at least one of the variables from $\{x,y,z\}$ will be true in $A'$.
      Hence, the assignment $A$, where $A(x)=A'(x)$ for all $x \in X$, is a satisfying assignment for $\phi$.

      To complete the proof we note that in $\phi'$, every variable occurs at most 4 times and each clause has exactly 3 literals. Moreover, the
      $\phi'$ can be constructed in polynomial time w.r.t. $n$ and $m$. We also note
      that given an arbitrary assignment, we can check whether it NAE-satisfies a CNF formula, in polynomial time.
    \end{proof}
  \end{appendices}
  
\end{document}